\documentclass[runningheads]{llncs}

\usepackage[dvipsnames]{xcolor}
\usepackage[boxed, vlined, linesnumbered]{algorithm2e}
\usepackage{amsfonts}
\usepackage[nosumlimits]{amsmath}
\usepackage{amssymb}
\usepackage{mathtools}
\usepackage{todonotes}
\usepackage{complexity}
\usepackage{hyperref}
\usepackage[capitalize]{cleveref}
\usepackage{dsfont}
\usepackage{cite}
\usepackage[pdftex]{adjustbox}
\usepackage{graphicx}
\usepackage{subcaption}
\usepackage{bbm}
\usepackage{dcolumn}
\usepackage{wrapfig}
\usepackage[title]{appendix}
\usepackage{wasysym}

\begin{document}


\newcommand{\lip}{longest induced path\xspace}
\newcommand{\lipsc}{\textsc{LongestInducedPath}\xspace}

\newcommand{\lpBase}{\ensuremath{\mathrm{LP}_{\mathrm{Base}}}\xspace}
\newcommand{\ilpBase}{\ensuremath{\mathrm{I}\lpBase}\xspace} 

\newcommand{\lpWalk}{\ensuremath{\mathrm{LP}_{\mathrm{Walk}}}\xspace}
\newcommand{\ilpWalk}{\ensuremath{\mathrm{I}\lpWalk}\xspace}

\newcommand{\lpCut}{\ensuremath{\mathrm{LP}_{\mathrm{Cut}}}\xspace}
\newcommand{\ilpCut}{\ensuremath{\mathrm{I}\lpCut}\xspace}

\newcommand{\lpFlow}{\ensuremath{\mathrm{LP}_{\mathrm{Flow}}}\xspace}
\newcommand{\ilpFlow}{\ensuremath{\mathrm{I}\lpFlow}\xspace}

\newcommand{\pCut}{\ensuremath{\mathcal{P}_{\mathrm{Cut}}}\xspace}
\newcommand{\pFlow}{\ensuremath{\mathcal{P}_{\mathrm{Flow}}}\xspace}
\newcommand{\ppFlow}{\ensuremath{\mathcal{P}'_{\mathrm{Flow}}}\xspace}

\newcommand{\Cnci}{\texttt{C}\ensuremath{^{\mathrm{n,c}}_{\mathrm{int}}}\xspace}
\newcommand{\Cncf}{\texttt{C}\ensuremath{^{\mathrm{n,c}}_{\mathrm{frac}}}\xspace}
\newcommand{\Cni}{\texttt{C}\ensuremath{^{\mathrm{n}}_{\mathrm{int}}}\xspace}
\newcommand{\Cnf}{\texttt{C}\ensuremath{^{\mathrm{n}}_{\mathrm{frac}}}\xspace}
\newcommand{\Cci}{\texttt{C}\ensuremath{^{\mathrm{c}}_{\mathrm{int}}}\xspace}
\newcommand{\Ccf}{\texttt{C}\ensuremath{^{\mathrm{c}}_{\mathrm{frac}}}\xspace}
\newcommand{\Ci}{\texttt{C}\ensuremath{_{\mathrm{int}}}\xspace}
\newcommand{\Cf}{\texttt{C}\ensuremath{_{\mathrm{frac}}}\xspace}
\newcommand{\Fnc}{\texttt{F}\ensuremath{^{\mathrm{n,c}}}\xspace}
\newcommand{\WW}{\texttt{W}\xspace}

\newcommand{\cvi}{\Cnci\xspace}
\newcommand{\nvi}{\Cni\xspace}
\newcommand{\cvf}{\Cncf\xspace}
\newcommand{\cei}{\Cci\xspace}
\newcommand{\cef}{\Ccf\xspace}
\newcommand{\cvm}{\Fnc\xspace}
\newcommand{\AcL}{\WW\xspace}

\newcommand{\dlt}{\delta^*}

\newcommand{\CRN}{\texttt{RWC}\xspace}
\newcommand{\MG}{\texttt{MG}\xspace}
\newcommand{\BAS}{\texttt{BAS}\xspace}
\newcommand{\BAL}{\texttt{BAL}\xspace}
\newcommand{\BA}{\texttt{BA}\xspace}

\newcommand{\unsc}{\adjustbox{scale={0.5}{1},raise={1pt}{\height}}{\textunderscore}}

\newcommand{\mysub}[1]{\subsubsection{#1}}


\bibliographystyle{splncs04}

\title{An Experimental Study of ILP Formulations for the Longest Induced Path Problem}
\titlerunning{Experimental Study of ILP Formulations for Longest Induced Path}

\author{Fritz B\"okler\orcidID{0000-0002-7950-6965} \and Markus Chimani\orcidID{0000-0002-4681-5550} \and Mirko~H.\ Wagner\orcidID{0000-0003-4593-8740} \and Tilo Wiedera\orcidID{0000-0002-5923-4114}}

\institute{Theoretical Computer Science, Osnabrück University, Germany
\email{\{fboekler,markus.chimani,mirwagner,tilo.wiedera\}@uni-osnabrueck.de}
}

\authorrunning{F.\ B\"okler, M.\ Chimani, M.\ H.\ Wagner, and T.\ Wiedera}

\maketitle

\begin{abstract}
	Given a graph $G=(V,E)$, the \lipsc problem asks for a maximum cardinality node subset $W\subseteq V$ such that the graph induced by $W$ is a path.
	It is a long established problem with applications, e.g., in network analysis.
	We propose novel integer linear programming (ILP) formulations for the problem and discuss efficient implementations thereof.
	Comparing them with known formulations from literature, we prove that they are beneficial in theory, yielding stronger relaxations.
	Moreover, our experiments show their practical superiority.
\end{abstract}

\section{Introduction}

Let $G=(V,E)$ be an undirected graph and $W\subseteq V$.
The \emph{$W\!$-induced graph} $G[W]$ contains exactly the nodes $W$ and those edges of
$G$ whose incident nodes are both in $W\!$. If $G[W]$ is a path,
	it is called an \emph{induced path}.
The length of a \lip is also referred to as the \emph{induced detour number} which was introduced more than 30 years ago~\cite{buckleyHarary}.
We denote the problem of finding such a path by \lipsc.
It is known to be \NP-complete, even on bipartite graphs~\cite{garey1979computers}.

The \lipsc problem has applications in molecular physics, analysis of social, telecommunication, and more general transportation networks~\cite{borgatti2013,jackson2010,barabasi2016,newman2010}~as well as pure graph and complexity theory.
It is closely related to the graph \emph{diameter}---the length of the longest among all shortest paths between any two nodes, which is a commonly analyzed communication property of social networks~\cite{matsypuraEtAl}.
A \lip witnesses the largest diameter that may occur by the deletion of any node subset in a node failure scenario~\cite{matsypuraEtAl}.
The \emph{tree-depth} of a graph is the minimum depth over all of its depth-first-search trees,
	and constitutes an upper bound on its treewidth~\cite{bodlaenderGilbertHafsteinssonKloks}, which is a well-established measure in parameterized complexity and graph theory.
Recently, it was shown that any graph class with bounded degree has bounded induced detour number iff it has bounded tree-depth~\cite{nesetrilOssonaDeMendez}.
Further, the enumeration of induced paths can be used to predict nuclear magnetic resonance~\cite{2014arXiv1404.7610U}.

\lipsc is not only \NP-complete, but also \W[2]-complete~\cite{cf2007}
and does not allow a polynomial $\mathcal{O}(|V|^{1/2-\epsilon})$-approximation, unless $\NP = \ZPP$~\cite{hastad1999,10.1007/BFb0028990}.
On the positive side, it can be solved in polynomial time for several graph classes, e.g.,
those of bounded mim-width (which includes interval, bi-interval, circular arc, and permutation graphs)~\cite{DBLP:journals/corr/abs-1708-04536} as well as $k$-bounded-hole, interval-filament, and other decomposable graphs~\cite{GAVRIL2002203}.
Furthermore, there are \NP-complete problems, such as \textsc{$k$-Coloring} for $k\geq 5$~\cite{GOLOVACH2014107} and
\textsc{Independent Set} \cite{LOZIN2003167}, that are polynomial on graphs with bounded induced detour number.

Recently the first non-trivial, general algorithms to solve the \lipsc problem exactly were devised by Matsypura et al.~\cite{matsypuraEtAl}.
There, three different integer linear programming (ILP) formulations were proposed: the first searches for a subgraph with largest diameter; the second utilizes properties derived from the average distance between two nodes of a subgraph;
the third models the path as a walk in which no shortcuts can be taken. Matsypura et al.\ show that the latter (see below for details) is the most effective in practice.

\mysub{Contribution.}
In \cref{sec:cutILP}, we propose novel ILP formulations based on cut and subtour elimination constraints.
We obtain strictly stronger relaxations than those proposed in~\cite{matsypuraEtAl} and describe a way to strengthen them even further in \cref{sec:polyhedral}.
After discussing some algorithmic considerations in \cref{sec:compFrame}, we show in \cref{sec:experiments} that our most effective models are also superior in practice.

\section{Preliminaries}
\label{sec:prelim}
\mysub{Notation.}
For~$k\in\mathbb N$, let $[k]\coloneqq \{0,\ldots,k-1\}$.
Throughout this paper, we consider a connected, undirected, simple graph $G=(V,E)$ as our input.
Edges are cardinality-two subsets of $V$.
If there is no ambiguity, we may write $uv$ for an edge~$\{u,v\}$.
Given a graph~$H$, we refer to its nodes and edges by $V(H)$ and $E(H)$, respectively.
Given a cycle $C$ in $G$, a \emph{chord} is an edge connecting two nodes of $V(C)$ that are not neighbors along~$C$.

\mysub{Linear programming (cf., e.g.,~\cite{schrijver}).}
A \emph{linear program} (LP) consists of a cost vector~$c\in\mathbb R^d$
	together with a set of linear inequalities, called \emph{constraints}, that define a polyhedron~$\mathcal P$ in $\mathbb R^d$.
We want to find a point~$x \in \mathcal P$ that maximizes the \emph{objective function}~$c^\intercal x$.
This can be done in polynomial time.
	Unless $\P=\NP$, this is no longer true when restricting $x$ to have integral components;
	the so-modified problem is an \emph{integer linear program} (ILP).
	Conversely, the \emph{LP relaxation} of an ILP is obtained by dropping the integrality constraints on the components of~$x$.
	The optimal value of an LP relaxation is a dual bound on the ILP's objective; e.g., an upper bound for maximization problems.
As there are several ways to \emph{model} a given problem as an ILP,
	one aims for models that yield small dimensions and strong dual bounds, 
to achieve good practical performance.
This is crucial,
	as ILP solvers are based on a branch-and-bound scheme that relies on iteratively solving LP relaxations to obtain dual bounds on the ILP's objective.
When a model contains too many constraints, it is often sufficient
	to use only a reasonably sized constraint subset to achieve provably optimal solutions.
This allows us to add constraints during the solving process, which is called \emph{separation}.
We say that model~$A$ is \emph{at least as strong} as model~$B$,
	if for all instances, the LP relaxation's value of model~$A$ is no further from the ILP optimum than that of~$B$.
If there also exists an instance for which $A$'s LP relaxation yields a tighter bound than that of~$B$, then~$A$ is \emph{stronger} than~$B$.

When referring to models, we use the prefix ``$\mathrm{ILP}$'' with an appropriate subscript.
When referring to their respective LP relaxations we write ``$\mathrm{LP}$'' instead.

\mysub{Walk-based model (state-of-the-art).}
Recently, Matsypura et al.~\cite{matsypuraEtAl} proposed an ILP model, \ilpWalk, that is the foundation of the fastest known exact algorithm (called \texttt{A3c} therein) for \lipsc.
They introduce timesteps, and for every node $v$ and timestep $t$ they introduce a variable that is $1$ iff $v$ is visited at time $t$.
Constraints guarantee that nodes at non-consecutive time points cannot be adjacent.
We recapitulate details in \cref{walk_appendix}.
Unfortunately, \ilpWalk yields only weak LP relaxations (cf.\ \cite{matsypuraEtAl} and \cref{sec:polyhedral}).
To achieve a practical algorithm, Matsypura et al.\ iteratively solve \ilpWalk for an increasing number of timesteps until the path found does not use all timesteps, i.e., a non-trivial dual bound is encountered.
In contrast to~\cite{matsypuraEtAl}, we consider the number of edges in the path (instead of nodes) as the objective value.

\section{New Models}\label{sec:cutILP}
We aim for models that exhibit stronger LP relaxations and are practically solvable via single ILP computations.
To this end, we consider what we deem a more natural variable space.
We start by describing a partial model \ilpBase, which by itself is not sufficient but constitutes the core of our new models.
To obtain a full model, \ilpCut, we add constraints that prevent subtours.

For notational simplicity, we augment $G$ to~$G^*\coloneqq (V^* \coloneqq V\cup\{s\},E^* \coloneqq E\cup\{ sv \}_{v\in V})$
	by adding a new node~$s$ that is adjacent to all nodes of $V$.
Within $G^*$, we look for a longest induced cycle through $s$, where we ignore induced chords incident to~$s$.
Searching for a cycle, instead of a path, allows us to homogeneously require that each \emph{selected} edge, i.e., edge in the solution, has exactly two adjacent edges that are also selected.
	Let $\dlt(e)\subset E^*$ denote the edges adjacent to edge~$e$ in $G^*$.
Each binary $x_e$-variable is $1$ iff edge $e$ is selected.
%
We denote the partial model below by \ilpBase:\!\!\!
\begin{subequations}
\begin{align}
	\max & \sum_{e \in E} x_e \label{eq:max}\\
	\text{s.t. }
	& \sum_{v \in V} x_{sv} = 2 \label{st:2s}\\
	& 2x_e \leq \sum_{f \in \dlt(e)} x_f \leq 2 && \forall e \in E \label{st:adj-edges} \\
	& x_e \in \{0,1\} && \forall e \in E^* \label{st:e_int}
\end{align}
\end{subequations}
Constraint~(\ref{st:2s}) requires to select exactly two edges incident with $s$.
To prevent chords, constraints~(\ref{st:adj-edges}) enforce that any (original) edge $e\in E$, even if not selected itself, is adjacent to at most two selected edges;
if $e$ is selected, precisely two of its adjacent edges need to be selected as well.

\mysub{Establishing connectivity.}
The above model is not sufficient: it allows for the solution to consist of multiple disjoint cycles, only one of which contains~$s$.
But still, these cycles have no chords in $G$, and no edge in $G$ connects any two cycles.
To obtain a longest single cycle $C$ through $s$---yielding the \lip $G[V(C)\setminus\{s\}]$---we thus have to forbid additional cycles in the solutions that
are not containing $s$.
In other words, we want to enforce that the graph induced by the $x$-variables is connected.

There are several established ways to achieve connectivity: 
To stay with \emph{compact} (i.e., polynomially sized) models, we could, e.g., augment \ilpBase
with Miller-Tucker-Zemlin constraints (which are known to be polyhedrally weak~\cite{bektasGouveia}) or multi-commodity-flow formulations (\ilpFlow; cf.~\cref{sec:flow}).
However, herein we focus on augmenting \ilpBase with \emph{cut} or \emph{(generalized) subtour elimination} constraints, resulting in the (non-compact) model we denote by \ilpCut, that is detailed later.
Such constraints are a cornerstone of many
	algorithms for diverse problems where they are typically superior (in particular in practice) than other known approaches~\cite{polzin2004, fischerSalazarGonzales, fischetti91}.
	While \ilpCut and \ilpFlow are polyhedrally equally strong (cf.\ \cref{sec:polyhedral}), 
	we know from other problems that the sheer size of the latter typically nullifies the potential benefit of its compactness. Preliminary experiments show that this is indeed the case here as well.

\mysub{Cut model (and generalized subtour elimination).}
Let $\dlt(W)\coloneqq\{ w \bar{w} \in E^* \mid w\in W, \bar{w}\in V^* \setminus W \}$ be the set of edges in the cut induced by $W \subseteq V^*$.
For notational simplicity, we may omit braces when referring to node sets of cardinality one.
We obtain \ilpCut by adding \emph{cut constraints} to \ilpBase:
\begin{subequations}
\begin{align}
	& \sum_{e \in \dlt(v)}x_e \leq \sum_{e \in \dlt(W)}x_e && \forall W \subseteq V, v \in W \label{st:cut}
\end{align}
These constraints ensure that if a node~$v$ is incident to a selected edge (by~(\ref{st:adj-edges}) there are then two such selected edges), any cut separating $v$ from $s$ contains at least two
	selected edges, as well. Thus, there are (at least) two edge-disjoint paths between $v$ and $s$ selected. Together with the cycle properties of \ilpBase, we can deduce that all selected edges form a common cycle through $s$.

An alternative view leads to \emph{subtour elimination constraints} $\sum_{e \in E : e \subseteq W} x_e \leq |W|-1$ for $W \subseteq V$, which
prohibit cycles not containing $s$ via counting.
It is well known that these constraints can be generalized using
binary node variables $y_v \coloneqq \tfrac12 \sum_{e \in \dlt(v)} x_e$ that indicate whether a node~$v\in V$ participates in the solution (in our case: in the induced path)~\cite{goemans}.
	\emph{Generalized subtour elimination constraints} thus take the form
\begin{align}
	& \sum_{e \in E : e \subseteq W} x_e \leq \sum_{w \in W \setminus \{v\}}y_w && \forall W \subseteq V, v \in W. \label{st:gsec}
\end{align}
\end{subequations}
One expects \ilpCut and ``\ilpBase with constraints~(\ref{st:gsec})'' to be equally strong as this is well-known for standard Steiner tree, and other related models~\cite{Goemans91, ChimaniKandyba10, Ljubic08}.
In fact, there even is a direct one-to-one correspondence between cut constraints~(\ref{st:cut}) and generalized subtour elimination constraints~(\ref{st:gsec}):
By substituting node-variables with their definitions in~(\ref{st:gsec}), we obtain
	$2 \sum_{e \in E : e \subseteq W} x_e \leq \sum_{w \in W \setminus \{v\}}\sum_{e \in \dlt(v)}x_e$.
	A simple rearrangement yields the corresponding cut constraint~(\ref{st:cut}).

\mysub{Clique constraints.}
We further strengthen our models by introducing additional inequalities.
Consider any clique (i.e., complete subgraph) in $G$.
The induced path may contain at most one of its edges to avoid induced triangles:
\begin{align}
	& \sum_{e \in E:e\subseteq Q}x_e \leq 1 && \forall Q \subseteq V \colon G[Q]\text{ is a clique} \label{st:cl_e}
\end{align}


\section{Polyhedral Properties of the LP Relaxations}
\label{sec:polyhedral}

We compare the above models w.r.t.\ the strength of their LP relaxations, i.e.,
	the quality of their dual bounds.
Achieving strong dual bounds is a highly relevant goal also in practice:
	one can expect a lower running time for the ILP solvers in case of better dual bounds since
	fewer nodes of the underlying branch-and-bound tree have to be explored.

Since \ilpWalk requires \emph{some} upper bound~$T$ on the objective value, we can only reasonably compare
this model to ours by assuming that we are also given this bound as an explicit constraint.
Hence, no dual bound of any of the considered models gives a worse (i.e., larger) bound than $T$.
As it has already been observed in \cite{matsypuraEtAl}, \lpWalk in fact \emph{always} yields this worst case bound:

\begin{proposition}\textnormal{\textbf{(Proposition~5 from~\cite{matsypuraEtAl})}}\label{thm:walk-value}
	For every instance and every number $T+1 \leq |V|$ of timesteps \lpWalk has objective value $T$.
\end{proposition}
\begin{proof}
	\label{sec:proofs}
	We set $x_v^t$ to $1/|V|$ for all $v\in V$ and $t\in [T+1]$.
	It is easy to see that this solution is feasible and attains the claimed objective value.
	\qed
\end{proof}

Note that \Cref{thm:walk-value} is independent of the graph.
Given that the \lip of a complete graph has length~$1$,
	we also see that the integrality gap of \ilpWalk is unbounded.
Furthermore, this shows that \ilpBase cannot be weaker than \ilpWalk.
We show that already the partial model \ilpBase is in fact \emph{stronger} than \ilpWalk.
Let therefore $\theta\coloneqq T-\mathrm{OPT}\in\mathbb N$, where $\mathrm{OPT}$ is the instance's (integral) optimum value.

\begin{proposition}\label{thm:cut-stronger}
	\ilpBase is stronger than \ilpWalk.
	Moreover, for every $\theta \geq 1$ there is an infinite family of instances on which \lpBase has objective value at most $\mathrm{OPT}+1$ and \lpWalk has objective value at least $T=\mathrm{OPT}+\theta$.
\end{proposition}
\begin{proof}
    By \cref{thm:walk-value}, \lpWalk will always attain value $T=\mathrm{OPT}+\theta$.
    To show the strength claim, it thus suffices to give instances where \lpBase yields a strictly tighter bound.

    Already a star with at least three leaves proves the claim, as 
    \lpBase guarantees a solution of optimal value $2$.
    However, it can be argued that such graphs and substructures are easy to
    preprocess. 
    Thus, we prove the claim with a more suitable instance class.

    Choose any $\ell\geq 3$, 
    start with two nodes $v_L,v_R$, connect them with $\ell$ internally node-disjoint paths of length 2, and add new node $v'$ with edge $v_Rv'$.
    A \lip in this graph 
    contains exactly $3$~edges: $v_Rv'$ and the two edges of one of the $v_L$-$v_R$-paths.
    Let $\deg(v) \coloneqq |\{e \in E : v \in E\}|$ denote the degree of node~$v$ in~$G$ without added star~$s$.
    By summing all constraints~(\ref{st:adj-edges}) we deduce
    \[2|E| \geq \sum_{e\in E}\sum_{f\in\dlt(e)} x_f \geq \underbrace{\sum_{e\in E}\sum_{f\in\dlt(e)\cap E} x_f}_{a} + \underbrace{\sum_{v\in V} \deg(v)\cdot x_{sv}}_b.\]
    For the double sum $a$ we see that any edge incident to $v_L$ or $v'$
    is considered $\ell$ times, i.e., it has $\ell$ adjacent edges, while the other edges are considered $\ell+1$ times. Thus $a\geq\ell \sum_{e\in E} x_e$.
    In the second sum $b$, $v_Rv'$ is the only edge with coefficient $1$ (instead of $\geq 2$), and we
    thus have $b\geq (2\sum_{v\in V} x_{sv}) - x_{sv'}$.
    By~(\ref{st:2s}) and the variable bounds we have $b\geq 4 - 1 = 3$.
    Since $|E|=2\ell+1$ we overall have
    $2(2\ell+1) \geq \ell \sum_{e\in E} x_e + 3$, giving objective value $\sum_{e\in E} x_e \leq 4-\tfrac1\ell$.
    As the objective must be integral, this even yields the optimal bound~$3$ when using \lpBase within an ILP solver.

    We furthermore note that, to achieve strictly two-connected graphs, we could, e.g.,
    also consider a cycle where each edge is replaced by two
    internally node-disjoint paths of length 2. However, in the above instance class the gap between the relaxations is larger, which is why we refrain from giving further details
    to the latter class.
	\qed
\end{proof}

Since \ilpCut only has additional constraints compared to \ilpBase
, this implies that \ilpCut is
 also stronger than \ilpWalk.
In fact, since constraints~\eqref{st:cut} cut off infeasible integral points contained in \ilpBase, \lpCut is clearly even a strict subset of \lpBase.
	As noted before, we can show that using a multi-commodity-flow scheme (cf.~\cref{sec:flow})
	results in LP relaxations equivalent to \lpCut:
	\begin{proposition}\label{thm:cut-eq-flow}
		\ilpFlow and \ilpCut are equally strong.
	\end{proposition}
	\begin{proof}
		Let \pCut and \pFlow be the polytope of \lpCut and \lpFlow, respectively. Let \ppFlow be the projection of \pFlow onto the $x$-variables
		by ignoring the $z$-variables.
		Then $\pCut=\ppFlow$.
		We show that the projection is surjective. Clearly, it retains the objective value.
	We observe that by constraints~(\ref{st:mcf-capacity}) for any node~$v$
			there can be at most $x_e$ units of flow along edge~$e$ that belong to some commodity~$v$.
		By constraint~(\ref{st:mcf-kirchhoff}), each node~$v \in V$
			sends $\sum_{e \in \dlt(v)}x_e$ units of flow that have to arrive at node~$s$.
		Consequently, the claim---both that any \lpFlow solution maps to an \lpCut solution and vice versa---follows directly from the duality of max-flow and min-cut.
	\qed
	\end{proof}

Let $\ilpCut^k$
denote \ilpCut
with clique constraints added for all cliques on at most~$k$ nodes.
We show that increasing the clique sizes yields a hierarchy of ever stronger models.
\begin{proposition}
\label{theo:cl_v_stronger}
	For any $k \geq 4$, $\ilpCut^k$ is stronger than $\ilpCut^{k-1}$.
\end{proposition}
\begin{proof}
	$\ilpCut^k$ is as least as strong as $\ilpCut^{k-1}$ as we only add new constraints.
	Let $G=K_k$, the complete graph on $k$ nodes.
	By choosing $Q = V$ in constraint~(\ref{st:cl_e}), $\lpCut^k$ has objective value $1$.

	However, $\lpCut^{k-1}$ allows a solution with objective value $\omega \coloneqq 1+\tfrac2{k-2} > 1$:
	We set $\tilde x_e \coloneqq \omega/\binom{k}{2}$ for each $e \in E$ and $\tilde x_{sv} \coloneqq \tfrac2k$ for each $v \in V$ to
		obtain an LP feasible solution~$\tilde x$ to $\lpCut^{k-1}$:
	Clearly, constraints~(\ref{st:2s},\ref{st:adj-edges}) are satisfied.
	The cut constraints~(\ref{st:cut}) are satisfied since edge variables are chosen uniformly (w.r.t.\ the two above edge types)
	and the right-hand side of the constraint sums over at least as many edge variables (per type) as the left-hand side.
	For any clique of size at most $k-1$, the left-hand side of its clique constraint~(\ref{st:cl_e}) sums up to at most $\binom{k-1}{2} \cdot \tilde x_e
	= \binom{k-1}{2}(1+\tfrac2{k-2})/\binom{k}{2}
	= 1$.

	We note that it is straight-forward to generalize $G$,
	so that it contains $K_k$ only as a subgraph,
	while retaining the property of
	having
	a gap between the two considered~LPs.
	\qed
\end{proof}

\section{Algorithmic Considerations}
\label{sec:compFrame}


\mysub{Separation.}
Since \ilpCut contains an exponential number of cut constraints~(\ref{st:cut}), it is not practical in its full form.
We follow the traditional separation pattern for branch-and-cut-based ILP solvers:
We initially omit cut constraints~(\ref{st:cut}), i.e., we start with model~$M\coloneqq\ilpBase$.
Iteratively, given a feasible solution to the LP relaxation of~$M$,
	we seek violated cut constraints and add them to $M$.
If no such constraints are found and the solution is integral, we have obtained a solution to \ilpCut.
Otherwise, we proceed by branching or---given a sophisticated branch-and-cut framework---by applying more general techniques.


Given an LP solution $\hat{x}$, we call an edge $e\in E$ \emph{active} if $\hat{x}_e > 0$.
Similarly, we say that a node is \emph{active}, if it has an active incident edge.
These active graph elements yield a subgraph $H$ of $G^*$.
For integral LP solutions, we simply compute the connected components of $H$ and
add a cut constraint for each component that does not contain $s$. We refer to this routine as \emph{integral separation}.
For a fractional LP solution, we compute the maximum flow value $f_v$ between $s$ and each active node $v$ in $H$;
the capacity of an edge $e\in E^*$ is equal to $\hat{x}_{e}$.
If $f_v<\sum_{e\in \dlt(v)} \hat{x}_{e}$, a cut constraint based on the induced minimum $s$-$v$-cut is added.
We call this routine \emph{fractional separation}.
Both routines manage to find a violated constraint if there is any, i.e., they are \emph{exact} separation routines.
In fact, this shows that an optimal solution to \lpCut can be computed in polynomial time~\cite{grotschelLovaszSchrijver}.
Note that already integral separation
suffices to obtain an exact, correct algorithm---we simply may need more branching steps than with fractional separation.

\mysub{Relaxing variables.}
As presented above, our models have $\Theta(|E|)$ binary variables, each of which may be used for branching by the ILP solver.
We can reduce this number, by introducing $\Theta(|V|)$ new binary variables $y_v$, $v\in V$,
that allow us to relax the binary $x_e$-variables, $e\in E$, to continuous ones.
The new variables are precisely those discussed w.r.t.\ generalized subtour elimination, i.e., we require $y_v = \tfrac{1}{2} \sum_{e \in \dlt(v)} x_e$.
Assuming $x_e$ to be continuous in $[0, 1]$, we have for every edge $e = \{v, w\} \in E\colon$ if $y_v = 0$ or $y_w = 0$ then $x_e = 0$.
Conversely, if $y_v = y_w = 1$ then $x_e = 1$ by~(\ref{st:adj-edges}).
Hence, requiring integrality for the $y$-variables (and, e.g., branching only on them), suffices to  ensure integral $x$ values.


\mysub{Handling clique constraints.}
We use a modified version of the Bron-Kerbosch algorithm \cite{EppsteinLoefflerStrash10} to list all maximal cliques.
For each such clique we add a constraint during the construction of our model.
Recall that there are up to $3^{n/3}$ maximal cliques \cite{Moon65}, but preliminary tests show that this effort is negligible compared to solving the ILP.
Thus, as our preliminary tests also show, other (heuristic) approaches of adding clique constraints to the initial model are not worthwhile.


\section{Computational Experiments}
\label{sec:experiments}

\mysub{Algorithms.}
We implement the best state-of-the-art algorithm, i.e., the \ilpWalk-based one by Matsypura et al.~\cite{matsypuraEtAl}.
We denote this algorithm by~``\WW''.
For our implementations of \ilpCut, we consider various parameter settings w.r.t.\ to the algorithmic considerations described in \cref{sec:compFrame}.
We denote the arising algorithms by ``\texttt{C}'' to which we attach sub- and superscripts defining the parameters:
the subscript ``$\mathrm{frac}$'' denotes that we use fractional separation in addition to integral separation. The superscript ``$\mathrm{n}$'' specifies that we introduce node variables as the sole integer variables.
The superscript ``$\mathrm{c}$'' specifies that we use clique constraints.
We consider all eight thereby possible \ilpCut implementations.


\mysub{Hard- and software.}
Our C\raisebox{0px}{\texttt{++}} (GCC 8.3.0) code uses SCIP 6.0.1~\cite{SCIP} as the Branch-and-Cut-Framework with CPLEX 12.9.0 as the LP solver.
We use OGDF snapshot-2018-03-28~\cite{OGDF} 
for the separation of cut constraints. We use igraph 0.7.1~\cite{igraph} to calculate all maximal cliques.
For \WW, we directly use CPLEX instead of SCIP as the Branch-and-Cut-Framework.
This does not give an advantage to our algorithms, since CPLEX is more than twice as fast as SCIP \cite{Achterberg2009} and we confirmed in preliminary tests that CPLEX is faster on \ilpWalk.
However, we use SCIP for our algorithms, as it allows better parameterizible user-defined separation routines.
We run all tests on an Intel Xeon Gold 6134 with 3.2\,GHz and 256\,GB RAM running Debian~9.
We limit each test instance to a single thread with a time limit of $20$ minutes and a memory limit of $8$\,GB.

\mysub{Instances.}
We consider the instances proposed for \lipsc in \cite{matsypuraEtAl} as well as  additional ones.
Overall, our test instances are grouped into four sets: \CRN, \MG, \BAS and \BAL.
The first set, denoted by \CRN, is a collection of 22 real-world networks, including communication and social networks of companies and of characters in books, as well as transportation, biological, and technical networks. See~\cite{matsypuraEtAl} for details on the selection.
The \emph{Movie Galaxy} (\MG) set consists of 773 graphs representing social networks of movie characters~\cite{mg}.
While \cite{matsypuraEtAl} considered only 17 of them, we use the full set here.
The other two sets are based on the Barab\'asi-Albert probabilistic model for scale-free networks~\cite{BA99}.
In~\cite{matsypuraEtAl}, only the chosen parameter values are reported, not the actual instances. Our set \BAS
recreates instances with the same values:
30~graphs for each choice $(|V|,d) \in \{(20,3), (30,3), (40,3), (40,2)\}$,
	where $|E| = (|V|-d) \cdot d$ describes the density of the graph.
As we will see, these small instances are rather easy for our models.
We thus also consider a set \BAL of graphs on 100 nodes; for each
	density~$d\in\{2,3,10,30,50\}$ we generate 30 instances.
See \url{http://tcs.uos.de/research/lip} for all instances, their sources, and detailed experimental results.

\mysub{Comparison to the state-of-the-art.}

\newcommand{\Fa}{\multicolumn{1}{@{\ }r<{\phantom{.}}@{}}{}}
\newcommand{\FT}{\!\clap{{\clock}}}
\newcommand{\FM}{\!\clap{\emph{M}}}
\newcommand{\bm}[1]{\multicolumn{1}{@{\ }r<{\textbf{.}}@{}}{\textbf{#1}}}
\newcommand{\bn}[1]{\textbf{#1}}
\begin{table}[tb]\centering
	\caption{Running times [s] on \CRN except for \texttt{yeast} and \texttt{622bus} (solved by none).
	We denote timeouts by \,~\FT~ and mark times within $5\%$ of the minimum in bold.}
	\label{tab:soc-net}
	\scalebox{0.779}{%
	\begin{tabular}{@{}l@{\,}|@{\,}rrr@{\,}|*{10}{@{\ }r<{.}@{}l@{\ }}}
instance	&	\multicolumn{1}{@{\,}l}{$\mathrm{OPT}$}	&	$|V|$	&	$|E|$	&	\multicolumn{2}{c}{\WW}	&	\multicolumn{2}{c}{\Ci}	&	\multicolumn{2}{c}{\Cf}	&	\multicolumn{2}{c}{\Cci}	&	\multicolumn{2}{c}{\Ccf}	&	\multicolumn{2}{c}{\Cni}	&	\multicolumn{2}{c}{\Cnf}	&	\multicolumn{2}{c}{\Cnci}	&	\multicolumn{2}{c}{\Cncf}	\\
\hline
		\texttt{high-tech}  & 13  & 33   & 91   & 15  & 40  & 0      & 90      & 1      & 11      &1 & 44 &	 	3 & 15 &    0      & 51      & 0       & 81      &\bm{0} & \bn{41} &    2 & 05    \\
		\texttt{karate}     & 9   & 34   & 78   & 2   & 98  & 1      & 73      & 1      & 65      &2 & 12 &	 	1 & 32 &    1      & 07      & 3       & 71      &\bm{0} & \bn{66} &    2 & 74    \\
		\texttt{mexican}    & 16  & 35   & 117  & 73  & 30  & 1      & 68      & 2      & 25      &1 & 12 &	 	3 & 59 &    1      & 22      & 1       & 34      &\bm{0} & \bn{87} &    0 & 99    \\
		\texttt{sawmill}    & 18  & 36   & 62   & 70  & 00  & 0      & 51      & \bm{0} & \bn{43} &0 & 50 &	 	\bm{0} & \bn{44} &    0      & 85      & 3       & 32      &0 & 82 &    3 & 34    \\
		\texttt{tailorS1}   & 13  & 39   & 158  & 83  & 80  & 4      & 78      & 7      & 92      &4 & 81 &	 	6 & 45 &    \bm{1}      & \bn{51}      & 1       & 87      &3 & 29 &    3 & 55    \\
\texttt{chesapeake} & 16  & 39   & 170  & 106 & 00  & \bm{1} & \bn{84} & 13     & 11      &2 & 11 &	 	11 & 00 &   2      & 29      & 4       & 88      &3 & 19 &    4 & 39    \\
		\texttt{tailorS2}   & 15  & 39   & 223  & 445 & 00  & 6      & 80      & 21     & 78      &11 & 92 &	 	14 & 91 &   3 & 20 & 4       & 31      &\bm{2} & \bn{89} &    3 & 14    \\
		\texttt{attiro}     & 31  & 59   & 128  & \Fa & \FT & 1      & 76      & 2      & 57      &2 & 48 &	 	1 & 75 &    1      & 20      & 1       & 75      &\bm{0} & \bn{89} &    1 & 19    \\
		\texttt{krebs}      & 17  & 62   & 153  & 522 & 00  & 3 & 86 & 28     & 21      &18 & 55 &	 	10 & 03 &   16     & 00      & 11      & 26      &3 & 90 &    \bm{2} & \bn{33}    \\
		\texttt{dolphins}   & 24  & 62   & 159  & \Fa & \FT & 7      & 95      & 27     & 59      &22 & 72 &	 	18 & 33 &   19     & 21      & \bm{2}  & \bn{99} &\bm{3} & \bn{01} &    4 & 70    \\
		\texttt{prison}     & 36  & 67   & 142  & \Fa & \FT & 13     & 36      & 5      & 87      &1 & 09 &	 	1 & 50 &    3      & 62      & 4       & 05      &\bm{1} & \bn{02} &    \bm{1} & \bn{02}    \\
		\texttt{huck}       & 9   & 69   & 297  & 41  & 70  & \Fa    & \FT     & 144    & 13      &19 & 46 &	 	42 & 22 &   114    & 27      & 11 & 63 &\bm{5} & \bn{96} &    7 & 49    \\
		\texttt{sanjuansur} & 38  & 75   & 144  & \Fa & \FT & 30     & 67      & 8      & 64      &24 & 86 &	 	10 & 33 &   8      & 22      & \bm{3}  & \bn{65} & \bm{3} & \bn{79} &    4 & 71    \\
		\texttt{jean}       & 11  & 77   & 254  & 121 & 00  & 464    & 89      & 52     & 89      &16 & 54 &	 	9 & 53 &    81     & 03      & 14      & 47      &\bm{3} & \bn{88} &    5 & 14   \\
		\texttt{david}      & 19  & 87   & 406  & \Fa & \FT & 666    & 25      & 719    & 46      &26 & 70 &	 	45 & 34 &   85     & 88      & 23      & 94      &\bm{6} & \bn{93} &    10 & 35   \\
\texttt{ieeebus}    & 47  & 118  & 179  & \Fa & \FT & 37     & 10      & 22     & 35      &39 & 82 &	 	10 & 60 &   15     & 69      & \bm{3}  & \bn{13} &22 & 72 &   5 & 61    \\
		\texttt{sfi}        & 13  & 118  & 200  & 44  & 40  & 47     & 41      & 4      & 39      &4 & 89 &	 	3 & 77 &    15     & 13      & 2  & 64 &3 & 31 &    \bm{2} & \bn{44}    \\
		\texttt{anna}       & 20  & 138  & 493  & \Fa & \FT & 21     & 58      & 296    & 69      &53 & 21 &	 	74 & 55 &   439    & 23      & 20      & 27      &\bm{7} & \bn{09} &    7 & 58    \\
		\texttt{usair}      & 46 & 332  & 2126 & \Fa & \FT & \Fa    & \FT     & \Fa    & \FT     &\Fa & \FT &	\Fa & \FT & \Fa    & \FT     & \Fa     & \FT     &\bm{922} & \bn{94} &  \Fa & \FT \\
		\texttt{494bus}     & 142 & 494  & 586  & \Fa & \FT & \Fa    & \FT     & 379    & 29      &\Fa & \FT &	379 & 97 &  \Fa    & \FT     & \bm{178}     & \bn{92}      &\Fa & \FT & \bm{170} & \bn{74}  \\
	\end{tabular}}
\end{table}

\begin{figure}[p]
	\begin{subfigure}{\textwidth}
		\includegraphics{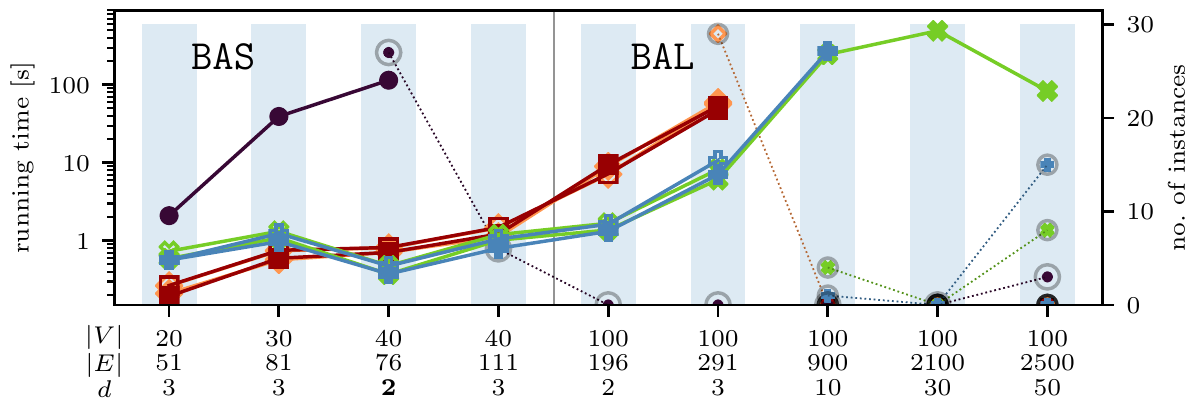}
			\caption{Running time on \BAS and \BAL
			\label{fig:ba}}
	\end{subfigure}

	\vspace*{1.5em}
	\begin{subfigure}{\textwidth}
		\begin{subfigure}{.47\textwidth}
			\begin{subfigure}{\textwidth}
				\includegraphics{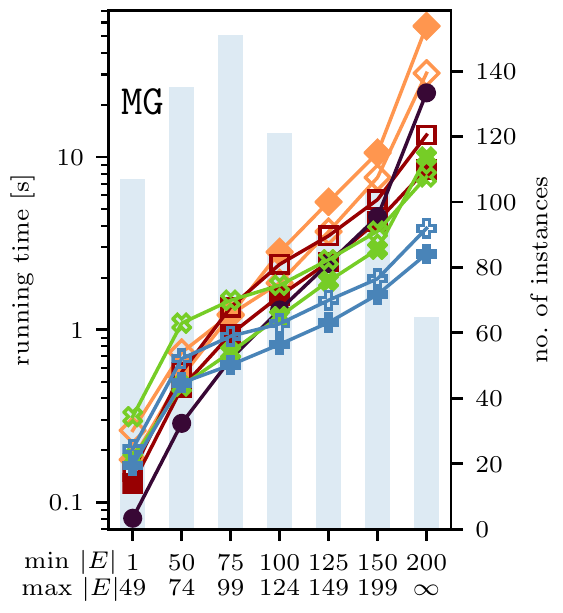}
				\caption{
						Running time on \MG
					\label{fig:mg}
					}
			\end{subfigure}
			\begin{subfigure}{\textwidth}
				\vspace{1.5em}
				\centering
				\includegraphics[width=.6\textwidth]{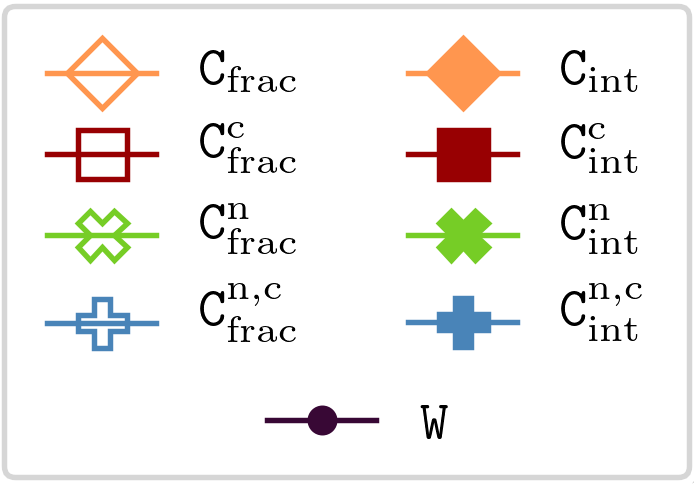}
			\end{subfigure}
		\end{subfigure}
		\begin{subfigure}{0.52\textwidth}
			\begin{subfigure}{\textwidth}
				\includegraphics{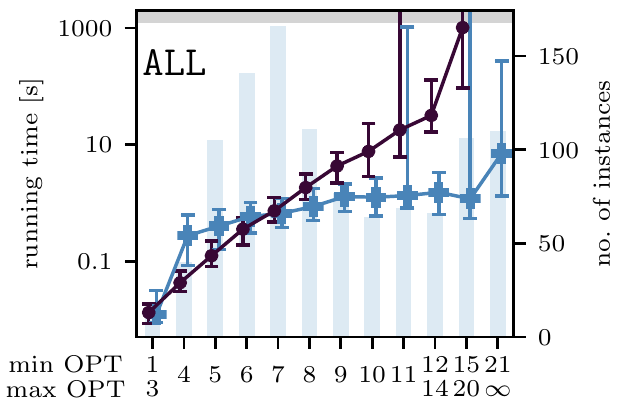}
				\caption{Running time vs.\ \textrm{OPT} (all instances)
					\label{fig:all}
				}
			\end{subfigure}
			\begin{subfigure}{\textwidth}
				\includegraphics{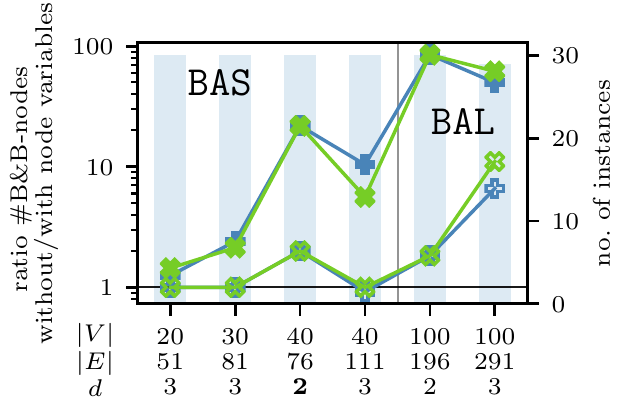}
				\caption{Reduction of B\&B-nodes by node var's on commonly solved \BAS and \BAL
					\label{fig:bnb}
				}
			\end{subfigure}
		\end{subfigure}
	\end{subfigure}

	\caption{
		Comparison between different ILP models.\\
			(a),(b): Each point is a median, where timeouts are treated as $\infty$ seconds.
			Bars in the background give the number of instances.
			Gray encircled markers, connected via dotted lines, show the number of solved instances (if not 100\%).\\
			(c): Whiskers mark the 20\% and 80\% percentile. The gray area marks timeouts.
		}
\end{figure}


We start with the most obvious question: Are the new models practically more effective than the state-of-the-art?
See \cref{fig:ba} for \BAS and \BAL, \cref{fig:mg} for \MG, and \cref{tab:soc-net} for \CRN.

We observe that on every benchmark set, the various \ilpCut implementations achieve the best running times and success rates.
The only exceptions are the instances from \MG (cf. \cref{fig:mg}): there, the overhead of the stronger model, requiring an explicit separation routine, does not pay off and \WW yields comparable performance to the weaker of the cut-based variants. 
On \BAS instances, 
the cut-based variants dominate (cf.\ \cref{fig:ba}):
	while all variants (detailed later) solve all of \BAS, \WW can only solve the instances for $d\in\{20,30\}$ reliably.
On \BAL (cf. \cref{fig:ba}) \WW 
fails on virtually all instances.
The cut-based model, however, allows implementations
	(detailed later) that solve all of these harder instances.
We point out one peculiarity on the \BAL instances, visible in \cref{fig:ba}.
The instances have 100 nodes but varying density.
As the density increases from 2 to 30, the median running times of all algorithmic variants increase and the median success rates decrease.
However, from $d=30$ to $d=50$ (where only \Cni is successful) the running times drop again and the success rate increases. 
Interestingly, the number of branch-and-bound (B\&B) nodes for $d=50$ is only roughly 1/7 of those for $d=30$. 
This suggests that the denser graphs may allow fewer (near-)optimal solutions and thus more efficient pruning of the search tree.

\mysub{Comparison of cut-based implementations.}
Choosing the best among the eight \ilpCut implementations is not as clear as the general choice of \ilpCut over
\ilpWalk.
In \cref{fig:ba}, \ref{fig:mg}, and \cref{tab:soc-net} we see that, while adding clique constraints is clearly beneficial on \MG, on \BAS and \CRN the benefit is less clear.
On \BAL, we do not see a benefit and for $d\in \{30, 50\}$ we even see a clear benefit of \emph{not} using clique constraints.
Each of the graphs from \BAL with $d\in \{30, 50\}$ has at least $4541$ maximal cliques---and therefore initial clique constraints---, whereas the \BAL graphs for $d=10$ and the \CRN graphs \texttt{yeast} and \texttt{usair} have at most $581$ maximal cliques and all other graphs have at most $102$.
The probably most surprising finding is the choice of the separation routine:
while the fractional variant is a quite fast algorithm and yields tighter dual bounds, the simpler integral separation performs better in practice.
This is in stark contrast to seemingly similar scenarios like TSP or Steiner problems, where the former is considered by default. In our case, the latter---being very fast and called more rarely---is seemingly strong enough to find effective cutting planes that allow the ILP solver to achieve its computations fastest.
This is particularly true when combined with the addition of node variables (detailed later).
In fact, \Cni is the only choice that can completely solve all large graphs in \BAL.

Adding node variables (and relaxing the integrality on the edge variables) nearly always pays off significantly (cf.\ \cref{fig:ba}, \ref{fig:mg}).
\cref{fig:bnb} shows that the models without node variables require many more B\&B-nodes.
In fact, looking more deeply into the data, \Cni requires roughly as few B\&B-nodes as \Cf without requiring the overhead of the more expensive separation routine.
Only for \BAS with $|V| \in \{20, 30\}$, the configurations without node variables are faster; on these instances, our algorithms only require $2$--$6.5$ B\&B-nodes (median).

%
%


\begin{figure}[tb]
	\centering
	\begin{subfigure}[t]{.575\textwidth}
		\begin{subfigure}[t]{.392\textwidth}
			\includegraphics{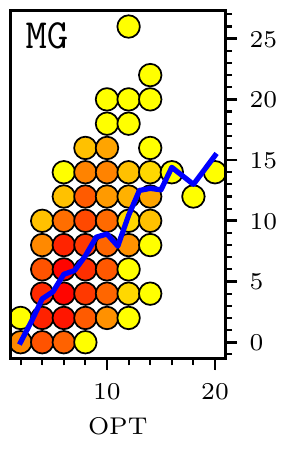}
		\end{subfigure}
		\begin{subfigure}[t]{0\textwidth}
			\includegraphics{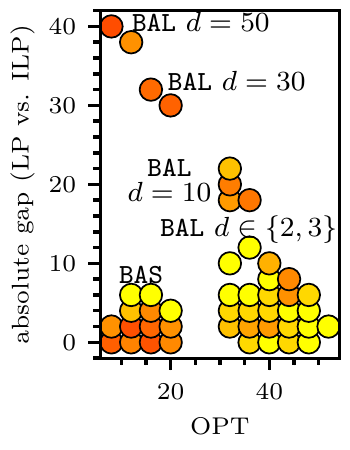}
		\end{subfigure}
		\caption{LP value vs.\ $\mathrm{OPT}$. Left: \MG; right: \BAS and \BAL.\label{fig:flame_vs_final}
		}
	\end{subfigure}
	\begin{subfigure}[t]{.025\textwidth}
		\vspace{-13.7em}
		\hspace{-1.95em}
		\includegraphics[scale=.72]{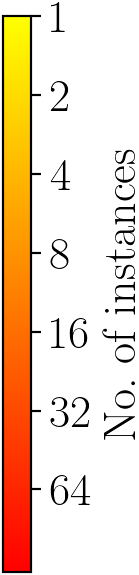}
		\phantom{x}
	\end{subfigure}
	\begin{subfigure}[t]{.330\textwidth}
		\includegraphics{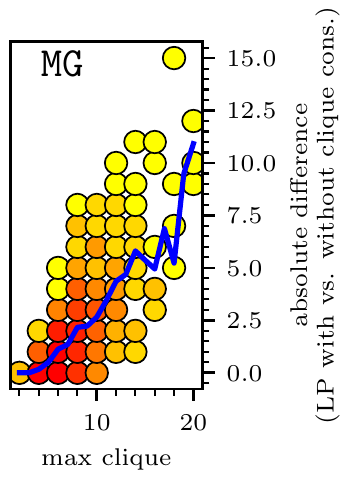}
		\caption{Maximal found clique size vs.\ LP value on \MG.\label{fig:flame_clique_mg}}
	\end{subfigure}\vspace{-.05em}
	\caption{Root LP relaxation of cut-based models. The blue line shows the median.}
\end{figure}

\mysub{Dependency of running time on the optimal value.}
Since the instances optimal value $\mathrm{OPT}$ determines the final size of the \ilpWalk instance, it is natural
to expect the running time of \WW to heavily depend on $\mathrm{OPT}$.
\cref{fig:all} shows that this is indeed the case.
The new models are less dependent on the solution size, as, e.g., witnessed by \Cnci in the same figure.

\mysub{Practical strength of the root relaxations.}
For our new models, we may ask how the integer optimal solution value and the value of the LP relaxation
	(obtained by any cut-based implementation with exact fractional separation) differ,
	see \cref{fig:flame_vs_final}.
The gap increases for larger values of $\mathrm{OPT}$. Interestingly, we observe that the \emph{density}
of the instance seems to play an important role: for \BAS and \BAL, the plot shows obvious clusters, which---without
a single exception---directly correspond to the different parameter settings as labeled. Denser graphs lead to
weaker LP bounds in general.

%
\cref{fig:flame_clique_mg} shows the relative improvement to the LP relaxation when adding clique constraints for \MG instances.
On the other hand for every instance of \BAS and \BAL  the root relaxation did not change by adding clique constraints.


\section{Conclusion} We propose new ILP models for \lipsc and prove that they yield stronger relaxations in theory than the previous state-of-the-art.
Moreover, we show that they
---generally, but also in particular in conjunction with further algorithmic considerations---clearly outperform all known approaches in practice.
We also provide strengthening inequalities based on cliques in the graph and prove that they form a hierarchy when increasing the size of the cliques.

It could be worthwhile to separate the proposed clique constraints (at least heuristically) to take advantage of their theoretical properties without overloading the initial model with too many such constraints.
As it is unclear how to develop an \emph{efficient} such separation scheme, we leave it as future research.




\bibliography{bib}

\begin{thebibliography}{10}
\providecommand{\url}[1]{\texttt{#1}}
\providecommand{\urlprefix}{URL }
\providecommand{\doi}[1]{https://doi.org/#1}

\bibitem{Achterberg2009}
Achterberg, T.: {SCIP:} solving constraint integer programs. Math. Prog.
  Comput.  \textbf{1}(1),  1--41 (2009)

\bibitem{BA99}
Barabási, A.L., Albert, R.: Emergence of {S}caling in {R}andom {N}etworks.
  Science  \textbf{286},  509--512 (1999)

\bibitem{barabasi2016}
Barabási, A.L.: {Network Science}. Cambridge University Press (2016)

\bibitem{bektasGouveia}
Bektaş, T., Gouveia, L.: {Requiem for the Miller-Tucker-Zemlin subtour
  elimination constraints?} EJOR  \textbf{236}(3),  820--832 (2014)

\bibitem{10.1007/BFb0028990}
Berman, P., Schnitger, G.: {O}n the {C}omplexity of {A}pproximating the
  {I}ndependent {S}et {P}roblem. Inf. Comput.  \textbf{96}(1),  77--94 (1992)

\bibitem{bodlaenderGilbertHafsteinssonKloks}
Bodlaender, H.L., Gilbert, J.R., Hafsteinsson, H., Kloks, T.: {Approximating
  Treewidth, Pathwidth, Frontsize, and Shortest Elimination Tree}. J. Alg.
  \textbf{18}(2),  238--255 (1995)

\bibitem{borgatti2013}
Borgatti, S.P., Everett, M.G., Johnson, J.C.: {Analyzing Social Networks}. SAGE
  Publishing (2013)

\bibitem{buckleyHarary}
Buckley, F., Harary, F.: On longest induced paths in graphs. Chinese Quart. J.
  Math.  \textbf{3}(3),  61--65 (1988)

\bibitem{arxivVersion}
Bökler, F., Chimani, M., Wagner, M.H., Wiedera, T.: {An Experimental Study of
  ILP Formulations for the Longest Induced Path Problem} (2020), {\tt
  arXiv:2002.07012 [cs.DS]}

\bibitem{cf2007}
Chen, Y., Flum, J.: {On Parameterized Path and Chordless Path Problems}. In:
  CCC. pp. 250--263 (2007)

\bibitem{OGDF}
Chimani, M., Gutwenger, C., Juenger, M., Klau, G.W., Klein, K., Mutzel, P.:
  {The} {Open} {Graph} {Drawing} {Framework} {(OGDF)}. In: Tamassia, R. (ed.)
  Handbook on Graph Drawing and Visualization, pp. 543--569. Chapman and
  Hall/CRC (2013), \url{www.ogdf.net}

\bibitem{ChimaniKandyba10}
Chimani, M., Kandyba, M., Ljubi\'{c}, I., Mutzel, P.: {Obtaining Optimal
  $k$-cardinality Trees Fast}. J. Exp. Alg.  \textbf{14},  5:2.5--5:2.23 (2010)

\bibitem{Ljubic08}
Chimani, M., Kandyba, M., Ljubi{\'{c}}, I., Mutzel, P.: {Strong Formulations
  for $2$-Node-Connected Steiner Network Problems}. In: COCOA. pp. 190--200.
  LNCS 5165 (2008)

\bibitem{igraph}
Csardi, G., Nepusz, T.: The igraph software package for complex network
  research. InterJournal, Complex Systems  \textbf{1695}, ~1--9 (2006),
  \url{http://igraph.sf.net}

\bibitem{EppsteinLoefflerStrash10}
Eppstein, D., L{\"o}ffler, M., Strash, D.: {Listing All Maximal Cliques in
  Sparse Graphs in Near-Optimal Time}. In: ISAAC. pp. 403--414. LNCS 6506
  (2010)

\bibitem{fischetti91}
Fischetti, M.: {Facets of two Steiner arborescence polyhedra}. Math. Prog.
  \textbf{51},  401--419 (1991)

\bibitem{fischerSalazarGonzales}
Fischetti, M., Salazar-Gonzalez, J., Toth, P.: {The Generalized Traveling
  Salesman and Orienteering Problems}. In: The Traveling Salesman Problem and
  Its Variations, Comb. Opt., vol.~12. Springer (2007)

\bibitem{garey1979computers}
Garey, M.R., Johnson, D.S.: Computers and Intractability: A Guide to the Theory
  of NP-Completeness. W. H. Freeman \& Co. (1979)

\bibitem{GAVRIL2002203}
Gavril, F.: Algorithms for maximum weight induced paths. Inf. Process. Let.
  \textbf{81}(4),  203--208 (2002)

\bibitem{SCIP}
Gleixner, A., Bastubbe, M., Eifler, L., Gally, T., Gamrath, G., Gottwald, R.L.,
  Hendel, G., Hojny, C., Koch, T., L{\"u}bbecke, M.E., Maher, S.J.,
  Miltenberger, M., M{\"u}ller, B., Pfetsch, M.E., Puchert, C., Rehfeldt, D.,
  Schl{\"o}sser, F., Schubert, C., Serrano, F., Shinano, Y., Viernickel, J.M.,
  Walter, M., Wegscheider, F., Witt, J.T., Witzig, J.: {The SCIP Optimization
  Suite 6.0}. ZIB-Report 18-26, Zuse Inst. Berlin (2018),
  \url{https://scip.zib.de}

\bibitem{goemans}
Goemans, M.X.: The steiner tree polytope and related polyhedra. Math. Prog.
  \textbf{63},  157--182 (1994)

\bibitem{Goemans91}
Goemans, M.X., soo Myung, Y.: {A Catalog of Steiner Tree Formulations}.
  Networks  \textbf{23},  19--28 (1993)

\bibitem{GOLOVACH2014107}
Golovach, P.A., Paulusma, D., Song, J.: Coloring graphs without short cycles
  and long induced paths. Disc. Appl. Math.  \textbf{167},  107--120 (2014)

\bibitem{grotschelLovaszSchrijver}
Gr{\"{o}}tschel, M., Lov{\'{a}}sz, L., Schrijver, A.: {Geometric Algorithms and
  Combinatorial Optimization}, Alg. and Comb., vol.~2. Springer (1988)

\bibitem{hastad1999}
H{\aa}stad, J.: Clique is hard to approximate within $n^{1 - \epsilon}$. Acta
  Math.  \textbf{182}(1),  105--142 (1999)

\bibitem{jackson2010}
Jackson, M.O.: {Social and Economic Networks}. Princeton University Press
  (2010)

\bibitem{DBLP:journals/corr/abs-1708-04536}
Jaffke, L., Kwon, O., Telle, J.A.: {P}olynomial-{T}ime {A}lgorithms for the
  {L}ongest {I}nduced {P}ath and {I}nduced {D}isjoint {P}aths {P}roblems on
  {G}raphs of {B}ounded {M}im-{W}idth. In: IPEC. pp. 21:1--13. LIPIcs 89 (2017)

\bibitem{mg}
Kaminski, J., Schober, M., Albaladejo, R., Zastupailo, O., Hidalgo, C.:
  {Moviegalaxies - Social Networks in Movies}. Harvard Dataverse (V3 2018)

\bibitem{LOZIN2003167}
Lozin, V., Rautenbach, D.: Some results on graphs without long induced paths.
  Inf. Process. Let.  \textbf{88}(4),  167--171 (2003)

\bibitem{matsypuraEtAl}
Matsypura, D., Veremyev, A., Prokopyev, O.A., Pasiliao, E.L.: On exact solution
  approaches for the longest induced path problem. EJOR  \textbf{278},
  546--562 (2019)

\bibitem{Moon65}
Moon, J.W., Moser, L.: {On Cliques in Graphs}. Israel J. of Math.
  \textbf{3}(1),  23--28 (1965)

\bibitem{nesetrilOssonaDeMendez}
Nesetril, J., de~Mendez, P.O.: {Sparsity - Graphs, Structures, and Algorithms},
  Alg. and Comb., vol.~28. Springer (2012)

\bibitem{newman2010}
Newman, M.: {Networks: An Introduction}. Oxford University Press (2010)

\bibitem{polzin2004}
Polzin, T.: {Algorithms for the Steiner problem in networks}. Ph.D. thesis,
  Saarland University, Saarbr{\"{u}}cken, Germany (2003)

\bibitem{schrijver}
Schrijver, A.: Theory of linear and integer programming. Wiley-Intersci. series
  in disc. math. and opt., Wiley (1999)

\bibitem{2014arXiv1404.7610U}
{Uno}, T., {Satoh}, H.: {An Efficient Algorithm for Enumerating Chordless
  Cycles and Chordless Paths}. In: Int. Conf. on Disc. Sci. pp. 313--324. LNCS
  8777 (2014)

\end{thebibliography}
\vfill

\pagebreak
\appendix
\section*{APPENDIX}
\label{sec:apx}
\section{Walk-Based Model (State-of-the-Art)}
	\label{walk_appendix}
The following ILP model, denoted by \ilpWalk,
	was recently presented in~\cite{matsypuraEtAl}.
It constitutes the foundation of the fastest known exact algorithm.
It models a timed walk through the graph that prevents ``short-cut'' edges.
Let $T$ denote an upper bound on the length of the path, i.e., on its number of edges.
For every node $v\in V$ and every point in time~$t \in [T+1]$ there is a variable~$x_v^t$ that is $1$ iff $v$ is visited at time~$t$~(\ref{st2:v}).
\begin{subequations}
\begin{align}
	\max~~~~&
	 \sum_{t=1}^{T} \sum_{v \in V} x_v^t  \label{eq2:max}\\
	 \text{s.t. }~~~~
	 &\sum_{v\in V} x_v^t \leq 1 && \forall t \in [T+1] \label{st2:step}\\
	 &\sum_{t=0}^{T} x_v^t \leq 1 && \forall v \in V \label{st2:visit}\\
	 &\sum_{v \in V} x_v^{t+1} \leq \sum_{v \in V} x_v^{t} && \forall t \in [T] \label{st2:ruption} \\
	 &x_v^t \leq 1 - \sum_{w \in V:vw\not\in E} x_w^{t+1} && \forall v \in V, t \in [T] \label{st2:nonadj}\\
	 &x_v^t \leq 1 - \sum_{\tau = t+2}^T x_w^\tau && \forall vw\in E, t \in [T-1] \label{st2:shortcut}\\
	 &x_v^t \in \{0,1\} && \forall v \in V, t \in [T+1] \label{st2:v}
\end{align}
In every step at most one node can be visited~(\ref{st2:step});
a node can be visited at most once~(\ref{st2:visit});
the time points have to be used consecutively~(\ref{st2:ruption});
nodes visited at consecutive
time points need to be adjacent~(\ref{st2:nonadj});
and nodes at non-consecutive time points
cannot be adjacent~(\ref{st2:shortcut}).

However, \ilpWalk yields only weak LP relaxations (cf.\ \cref{sec:polyhedral}).
To obtain a practical algorithm, the authors of \cite{matsypuraEtAl}
	iteratively solve \ilpWalk for increasing values of~$T$ until its optimal objective value becomes less than $T$.
They use the graph's diameter as a lower bound on $T$ to avoid trivial calls.
In addition, they add supplemental symmetry breaking inequalities.
\end{subequations}

\section{Multi-Commodity-Flow Model}
	\label{sec:flow}
	A flow formulation allows a \emph{compact}, i.e., polynomially-sized, model.
	We start with \ilpBase and extend it in the following way:
	Each node $v \in V$ is assigned a commodity and sends---if $v$ is part of the induced path---two units of flow of this commodity from $v$ to $s$ using only selected edges, where edges have capacity one (per commodity).
	This ensures that each node in the solution lies on a common cycle with~$s$.
	Consider the bidirected arc set~$A^*\coloneqq \{ (vw), (wv) \mid \{v,w\}\in E^*\}$ that consists of a directed arc for both directions of each edge in $E^*$.
	Let $\dlt_\mathrm{out}(v)$ ($\dlt_\mathrm{in}(v)$) denote the arcs of $A^*$ with source (resp.\ target) $v\in V$.
	We use variables $z_a^v$ to model the flow of commodity~$v$ over arc~$a\in A^*$; we do not actively require them to be binary.
	The below model, together with \ilpBase, forms \ilpFlow.
	\begin{subequations}
	\begin{align}
		& z^v_{(uw)} \leq x_{\{u,w\}} && \forall v\in V, (uw)\in A^* \label{st:mcf-capacity} \\
		& \sum_{a \in \dlt_\mathrm{out}(w)}z^v_a = \sum_{a \in \dlt_\mathrm{in}(w)}z^v_a + \mathbbm{1}_{w=v} \cdot \sum\limits_{e\in\dlt(v)} x_e  && \forall w,v \in V \label{st:mcf-kirchhoff} \\
		& 0\leq z^v_a \leq 1 && \forall v \in V, a \in A^*
	\end{align}
	\end{subequations}
	The capacity constraints~(\ref{st:mcf-capacity}) ensure that flow is only sent over selected edges.
	Equations~(\ref{st:mcf-kirchhoff}) model flow preservation (up to, but not including, the sink $s$) and send the commodities away from their source $v$, if $v$ is part of the solution.

\end{document}